\documentclass[sigconf, nonacm]{acmart}
\usepackage{tikz}
\usetikzlibrary{calc,patterns,arrows.meta,decorations.pathreplacing}
\usepackage{amsmath,mathtools,booktabs,graphicx}
\usepackage{algorithm,algpseudocode}
\newtheorem{theorem}{Theorem}
\newtheorem{lemma}[theorem]{Lemma}
\newtheorem{proposition}[theorem]{Proposition}
\newcommand{\Ts}{T^{*}}

\newcommand{\ZN}{\mathbb{Z}_N}
\AtBeginDocument{%
  }

\begin{document}

\title{PASCAL: A Progress Divergence-Aware Shared-Cache Model}

\author{Zhongchun Zhou}
\email{zzhouch@connect.ust.hk}
\affiliation{%
  \institution{The Hong Kong University of Science and Technology}
  \streetaddress{Clear Water Bay}
  \city{Kowloon}
  \country{Hong Kong}
  }
\orcid{0009-0000-7037-7418}

\author{Ya Wang}
\email{ywangmu@connect.ust.hk}
\affiliation{%
  \institution{The Hong Kong University of Science and Technology}
  \streetaddress{Clear Water Bay}
  \city{Kowloon}
  \country{Hong Kong}
  }
\orcid{0009-0009-1114-7159}

\author{Songtao Mao}
\email{smao13@jhu.edu}
\affiliation{%
  \institution{Johns Hopkins University}
  \streetaddress{Charles Street}
  \city{Baltimore}
  \state{MD}
  \country{USA}
  }
\orcid{0000-0002-8292-3161}

\author{Chengtao Lai}
\email{claiaf@connect.ust.hk}
\affiliation{%
  \institution{The Hong Kong University of Science and Technology}
  \streetaddress{Clear Water Bay}
  \city{Kowloon}
  \country{Hong Kong}
  }
\orcid{0000-0002-9547-9653}
\author{Wei Zhang}
\authornote{Corresponding author}
\email{eeweiz@ust.hk}
\orcid{0000-0002-7622-6714}
\affiliation{%
  \institution{The Hong Kong University of Science and Technology}
  \streetaddress{Clear Water Bay}
  \city{Kowloon}
  \country{Hong Kong}}
\newcommand{\songtaonote}[1]{\textcolor{red}{Songtao's note: #1}}
\renewcommand{\shortauthors}{Zhou et al.}
\newcommand{\CIGemmGBHigh}{8.8}
\newcommand{\CIGemmGBLow}{4.2}
\newcommand{\CIGemmThorHigh}{6.5}
\newcommand{\CIGemmThorLow}{2.2}
\newcommand{\CIOverallPascalHigh}{15.5}
\newcommand{\CIOverallPascalLow}{9.3}
\newcommand{\CIOverallSDCMHigh}{30.9}
\newcommand{\CIOverallSDCMLow}{21.8}
\newcommand{\CIOverallTileHigh}{41.7}
\newcommand{\CIOverallTileLow}{30.0}
\newcommand{\CIScanGBHigh}{32.3}
\newcommand{\CIScanGBLow}{9.6}
\newcommand{\CIScanThorHigh}{22.1}
\newcommand{\CIScanThorLow}{14.5}
\newcommand{\CostGBFitSec}{0.09}
\newcommand{\CostGBPanelPredictSec}{0.22}
\newcommand{\CostGBPanelProfileMin}{17}
\newcommand{\CostGBPredictMs}{0.50}
\newcommand{\CostGBProfileSec}{2.5}
\newcommand{\CostGBRefKernelMin}{2.0}
\newcommand{\CostGBRefRuns}{756}
\newcommand{\CostGBRefWallMin}{29}
\newcommand{\CostGemmGBFitMs}{7}
\newcommand{\CostGemmGBPredictMs}{0.8}
\newcommand{\CostGemmGBRefWallMin}{5}
\newcommand{\CostGemmThorFitMs}{7}
\newcommand{\CostGemmThorPredictMs}{0.8}
\newcommand{\CostGemmThorRefWallMin}{15}
\newcommand{\CostHost}{Apple M2 Pro}
\newcommand{\CostPipelineUs}{15}
\newcommand{\CostThorFitSec}{0.09}
\newcommand{\CostThorPanelPredictSec}{0.22}
\newcommand{\CostThorPanelProfileMin}{52}
\newcommand{\CostThorPredictMs}{0.53}
\newcommand{\CostThorProfileSec}{7.4}
\newcommand{\CostThorRefKernelMin}{4.7}
\newcommand{\CostThorRefRuns}{756}
\newcommand{\CostThorRefWallMin}{92}
\newcommand{\DeltaGemmGBML}{0.19}
\newcommand{\DeltaGemmGBMLHigh}{2.09}
\newcommand{\DeltaGemmGBMLLow}{$-$1.89}
\newcommand{\DeltaGemmThorML}{2.90}
\newcommand{\DeltaGemmThorMLHigh}{5.96}
\newcommand{\DeltaGemmThorMLLow}{0.02}
\newcommand{\DeltaOverallSDCM}{14.1}
\newcommand{\DeltaOverallSDCMHigh}{19.3}
\newcommand{\DeltaOverallSDCMLow}{9.1}
\newcommand{\DeltaOverallTileSight}{23.7}
\newcommand{\DeltaOverallTileSightHigh}{30.1}
\newcommand{\DeltaOverallTileSightLow}{17.6}
\newcommand{\DeltaScanGBML}{3.85}
\newcommand{\DeltaScanGBMLHigh}{12.94}
\newcommand{\DeltaScanGBMLLow}{$-$5.82}
\newcommand{\DeltaScanThorML}{1.19}
\newcommand{\DeltaScanThorMLHigh}{4.20}
\newcommand{\DeltaScanThorMLLow}{$-$1.95}
\newcommand{\EtoEGemmGBBoundGapHigh}{31.5}
\newcommand{\EtoEGemmGBBoundGapLow}{12.9}
\newcommand{\EtoEGemmGBBoundMeasured}{10.11}
\newcommand{\EtoEGemmGBBoundPascal}{19.85}
\newcommand{\EtoEGemmGBBoundRuns}{27}
\newcommand{\EtoEGemmGBBoundTile}{42.06}
\newcommand{\EtoEGemmGBGapHigh}{8.5}
\newcommand{\EtoEGemmGBGapLow}{1.9}
\newcommand{\EtoEGemmGBHiddenMeasured}{11.08}
\newcommand{\EtoEGemmGBHiddenPascal}{11.13}
\newcommand{\EtoEGemmGBHiddenTile}{10.93}
\newcommand{\EtoEGemmGBMeasured}{10.92}
\newcommand{\EtoEGemmGBPascal}{13.14}
\newcommand{\EtoEGemmGBTile}{18.17}
\newcommand{\EtoEGemmThorBoundGapHigh}{6.4}
\newcommand{\EtoEGemmThorBoundGapLow}{1.2}
\newcommand{\EtoEGemmThorBoundMeasured}{78.97}
\newcommand{\EtoEGemmThorBoundPascal}{80.08}
\newcommand{\EtoEGemmThorBoundRuns}{20}
\newcommand{\EtoEGemmThorBoundTile}{83.73}
\newcommand{\EtoEGemmThorGapHigh}{0.9}
\newcommand{\EtoEGemmThorGapLow}{$-$0.6}
\newcommand{\EtoEGemmThorHiddenMeasured}{82.05}
\newcommand{\EtoEGemmThorHiddenPascal}{82.03}
\newcommand{\EtoEGemmThorHiddenTile}{81.46}
\newcommand{\EtoEGemmThorMeasured}{81.54}
\newcommand{\EtoEGemmThorPascal}{81.70}
\newcommand{\EtoEGemmThorTile}{81.84}
\newcommand{\EtoEOverallMeasured}{48.14}
\newcommand{\EtoEOverallPascal}{48.73}
\newcommand{\EtoEOverallTile}{50.02}
\newcommand{\EtoEScanGBBoundRuns}{0}
\newcommand{\EtoEScanGBMeasured}{23.89}
\newcommand{\EtoEScanGBPascal}{23.89}
\newcommand{\EtoEScanGBTile}{23.89}
\newcommand{\EtoEScanThorBoundRuns}{2}
\newcommand{\EtoEScanThorMeasured}{76.20}
\newcommand{\EtoEScanThorPascal}{76.21}
\newcommand{\EtoEScanThorTile}{76.21}
\newcommand{\ErrGBDepthA}{63.0}
\newcommand{\ErrGBDepthB}{11.6}
\newcommand{\ErrGBDepthC}{0.2}
\newcommand{\ErrGBDepthD}{6.6}
\newcommand{\ErrGBDepthE}{29.6}
\newcommand{\ErrGBDepthF}{12.0}
\newcommand{\ErrGBDepthG}{12.9}
\newcommand{\ErrGBMedian}{6.0}
\newcommand{\ErrGBWOne}{1.25}
\newcommand{\ErrGBWithoutWorst}{10.5}
\newcommand{\ErrGBWorstShare}{50}
\newcommand{\ErrThorDepthA}{7.4}
\newcommand{\ErrThorDepthB}{11.2}
\newcommand{\ErrThorDepthC}{22.2}
\newcommand{\ErrThorDepthD}{17.0}
\newcommand{\ErrThorDepthE}{27.9}
\newcommand{\ErrThorDepthF}{18.2}
\newcommand{\ErrThorDepthG}{23.1}
\newcommand{\ErrThorMedian}{15.0}
\newcommand{\ErrThorWOne}{0.55}
\newcommand{\ErrThorWithoutWorst}{15.9}
\newcommand{\ErrThorWorstShare}{19}
\newcommand{\NcuGBSec}{2.2}
\newcommand{\NcuSpeedupGB}{4{,}300}
\newcommand{\NcuSpeedupThor}{13{,}600}
\newcommand{\NcuThorSec}{7.3}
\newcommand{\PGemmGBML}{0.027}
\newcommand{\PGemmGBSDCM}{1.2e-08}
\newcommand{\PGemmGBTileSight}{3.5e-09}
\newcommand{\PGemmThorML}{0.00046}
\newcommand{\PGemmThorSDCM}{1e-07}
\newcommand{\PGemmThorTileSight}{4e-08}
\newcommand{\PScanGBML}{0.078}
\newcommand{\PScanGBSDCM}{0.0051}
\newcommand{\PScanGBTileSight}{3.4e-05}
\newcommand{\PScanThorML}{0.19}
\newcommand{\PScanThorSDCM}{0.0055}
\newcommand{\PScanThorTileSight}{0.0038}
\newcommand{\PlainGBSec}{0.7}
\newcommand{\PlainThorSec}{0.9}
\newcommand{\SimLargestMin}{64}
\newcommand{\SimLargestMma}{45}
\newcommand{\SimLargestN}{2048}
\newcommand{\SimOneHigh}{143}
\newcommand{\SimOneLow}{52}
\newcommand{\SimPanelKHours}{67.5}
\newcommand{\SimRuns}{12}
\newcommand{\WinGemmGBML}{70}
\newcommand{\WinGemmGBSDCM}{85}
\newcommand{\WinGemmGBTileSight}{88}
\newcommand{\WinGemmThorML}{75}
\newcommand{\WinGemmThorSDCM}{83}
\newcommand{\WinGemmThorTileSight}{87}
\newcommand{\WinScanGBML}{67}
\newcommand{\WinScanGBSDCM}{81}
\newcommand{\WinScanGBTileSight}{81}
\newcommand{\WinScanThorML}{60}
\newcommand{\WinScanThorSDCM}{71}
\newcommand{\WinScanThorTileSight}{74}

\begin{abstract}
In modern AI accelerators and GPUs, many concurrent cores repeatedly access the same shared data. This pattern occurs in attention, where different query (Q) tiles share the same key and value (K/V) blocks, GEMM, where every tile in a row reads the same slice, and many other operators. We name this pattern shared cyclic scan. Due to a significant amount of data reuse in this pattern, the cache is expected to capture as much data reuse as possible and largely reduce requests sent to the main memory for both performance and energy consumption concerns. However, in reality, because of the intrinsic asynchrony of multi-cores, the actual cache miss rate and DRAM traffic can be much higher compared to ideal cases. In this paper, we propose PASCAL, a shared-cache model for shared cyclic scans. It calibrates finite-run traffic, which reflects progress divergence, on reference configurations and interpolates it along static program structure such as occupancy to predict the miss rate before execution. 
PASCAL supports software configuration exploration without target traces or counters at scales where cycle-accurate simulation is impractical, while its analysis gives a sharp $2\sigma-1$ sufficient capacity condition for preserving LRU sharing. 
Across held-out scans and GEMM on GB10 and Thor, PASCAL reaches 12.04\% balanced fill-equivalent miss-rate MAPE.
Replacing TileSight's cache component with PASCAL lowers GB10 GEMM latency MAPE from 18.17\% to 13.14\%.
\end{abstract}




\maketitle
\section{Introduction}
Accurate performance prediction for deep learning applications can benefit both operator developers and hardware architects. While some metrics of an AI chip, including latencies and throughput of some components, are highly stable, predicting cache miss rates remains formidable due to: (1) the dynamicity and structural complexity of a cache and its replacement policy, and (2) the runtime variations in scheduling policies, inter-core memory request interleaving, and arbitration. Although some designs comprise only scratchpad memories (SPMs), a shared cache remains a common choice in NVIDIA and AMD GPUs and HUAWEI Ascend series~\cite{davinci}.

Among the factors above, inter-core memory request interleaving can cause a remarkable difference in cache-DRAM traffic. If requests to the same cache line from different cores are issued close enough in time, they can be captured by the L2 cache or some other structures, like the L2 Request Coalescer in NVIDIA GPUs. But if they are issued far apart in time, a cache miss may happen, leading to excessive cache-DRAM traffic, which can be several times the traffic of an ideal-case estimate. This phenomenon, caused by the intrinsic asynchrony of a multi-core processor, has been observed in attention kernels~\cite{simfa}.
However, current analytical models for deep learning, such as TileSight~\cite{tilesight} and LCM~\cite{lcm}, calculate the reuse distance profile with a static assumption. This can underestimate the cache miss rate significantly.

In contrast, our method measures traffic at a set of reference configurations and interpolates between them to predict an unmeasured configuration. To summarize, we make the following contributions:
\begin{itemize}
\item For the TTL (time-to-live) cache, we specialize the working set theory for shared cyclic scans, a memory access pattern common in AI workloads. We derive a closed-form formula to calculate exact miss counts.
\item We extend the theory beyond the TTL cache. For the LRU cache, a bound on its miss-rate deviation from the TTL cache is given. Additionally, a fractional residency-budget bound
applies to every demand-paging policy, including offline
Belady.
\item In our experiments, progress divergence across processor cores raises the cache miss rate and DRAM traffic several-fold. Prefetch depth, compute work, occupancy, and run length all affect it, and neither occupancy nor initial memory pressure alone predicts it. 
\item We propose a pipeline for the prediction of the cache miss rate. Considering the dynamic features of the progress divergence, it achieves a MAPE of 12.04\% and outperforms the main baselines. End-to-end experiments further prove the feasibility of connecting this model to a complete AI performance prediction pipeline.
\end{itemize}

\section{Background and Motivation}
\label{sec:background}
\subsection{Existing Cache Models}
Existing cache predictors obtain locality information in three main ways.
First, frequency-based stochastic models such as Che--IRM infer a steady-state
miss rate from request laws and cache capacity~\cite{che1,che2,che3}. They
cannot distinguish our progress-aligned and divergent executions: every worker has
the same marginal address frequency, while the number of fills per round can
range from 1 to $M$ because requests from different workers may share a
fill. Second, working-set, footprint, and stack-distance methods characterize
an observed access order~\cite{denning1,denning2,stack,hotl}; SDCM further
maps reuse distance to hit probability under an associativity and
random-mapping model~\cite{brehob1999}. Statistical sampling can reduce the
measurement cost~\cite{statcache,statstack}, but some source must still
provide the target execution's interleaving or reuse profile. Third, static
program analyses show that reuse profiles can sometimes be inferred before
execution~\cite{whole,array}. For freely progressing GPU workers, however,
the logical per-worker order does not determine the cross-worker order:
occupancy, generated memory instructions, prefetch depth, and intervening
computation change relative progress during execution. Tay and Zou developed a parameterized page-fault equation that relates memory size to page faults through a conjectured invariant capturing the interaction between reference behavior and replacement policy~\cite{tay2006}. PASCAL addresses a complementary question: how relative progress among concurrent scans affects shared-cache traffic at a given capacity?

\subsection{Motivation}
As mentioned earlier, even when the elements and the access order are identical for each worker, multiple factors may cause a large difference in cache-DRAM traffic. All workers start from the same initial condition, but in some cases, there can be a large difference in their progress after running for some time: some workers have finished a large amount of work, while others have only finished much less. This difference in progress in turn affects the number of reuses in the shared cache. Hence the motivation of this work is to \underline{figure out the dynamics in this system}: as the system executes, how does the progress of the cores diverge from each other? How does the number of cache misses change over time?
We therefore first present an exact solution to the fixed-phase geometry, and then use one reusable calibration per device and static target-code features to predict the accumulated traffic of an unseen dynamic configuration, without profiling that target.
\section{Static Geometry and Policy Bounds}
\label{sec:static}
This section asks two questions about a shared cyclic scan with $M$ workers.
First, if the workers' relative positions are given, how much memory
traffic results, and could a better replacement policy reduce it?
Sections~\ref{sec:model}--\ref{sec:policy} answer with an exact count
(Theorem~\ref{thm:ttl}) and a bound that holds for every policy
(Theorem~\ref{thm:policy}). 
Second, when do separated workers retain sharing, and when can cache hits
close a timing gap? Theorems~\ref{thm:coherence}
and~\ref{thm:feedback} answer this in two steps. When the workers
advance at their own pace, how far ahead may a leader run and still
leave the blocks it fetched in the cache for the others?
Theorem~\ref{thm:coherence} gives a sharp capacity threshold. Once a
worker has fallen behind, does hitting on the blocks the leaders fetched
let it catch up? Under the modeled hidden-latency condition,
Theorem~\ref{thm:feedback} shows that hits cannot speed it up. 
Appendix~\ref{app:notation} lists the notation, and Appendix~\ref{app:sec3} proves all lemmas and theorems.
\subsection{Model}
\label{sec:model}

Let $\ZN=\{0,\dots,N-1\}$ index the data blocks (e.g. tiles) of one memory access stream. The cache capacity $C$ is therefore measured in blocks. There are $M$ concurrent workers (e.g. GPU thread blocks), and worker $w$ carries a \emph{phase} $s_w\in\ZN$. Time advances in \emph{rounds}; in round $k$ worker $w$ accesses block:
\begin{equation}
a_w(k) \;=\; (k+s_w) \bmod N .
\label{eq:trace}
\end{equation}
A worker with a larger phase is \emph{ahead}: if $g=(s_{w'}-s_w)\bmod N$ is small
and positive, worker $w'$ touched worker $w$'s current block $g$ rounds ago, so
$w'$ is a \emph{leader} of $w$ at lag $g$. The global trace is the round-by-round
interleaving of the $M$ workers; within a round we fix an arbitrary but
consistent order and index workers accordingly. 

Let $y_1<\cdots<y_D$ be the distinct phase values, where multiple
workers may share a phase. Set $y_{D+1}=y_1+N$, and define the cyclic
gap vector $\boldsymbol g=(g_1,\ldots,g_D)$ by
\[
    g_i=y_{i+1}-y_i,\qquad 1\le i\le D.
\]
These gaps are positive and satisfy $\sum_{i=1}^D g_i=N$.
The definition includes the single-phase case, for which $g_1=N$.
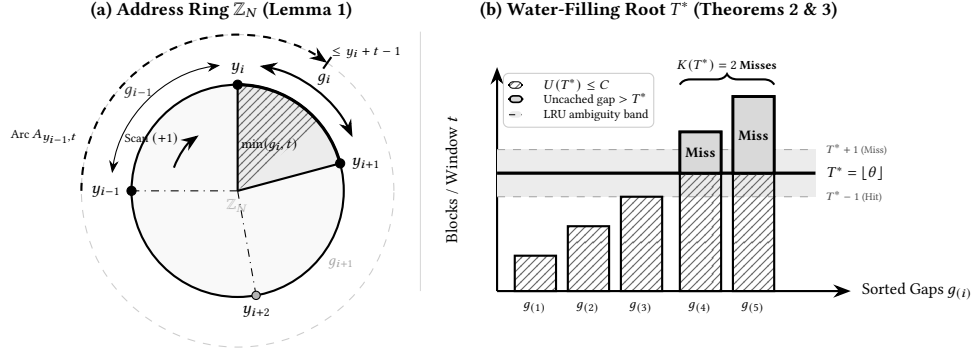
\begin{figure*}[t]
\centering
\begin{tikzpicture}[>=Stealth, font=\small, scale=0.78, every node/.style={transform shape}]

  \node[anchor=south west, font=\normalsize\bfseries] at (-2.6, 4.7) 
    {(a) Address Ring $\mathbb{Z}_N$ (Lemma~\ref{lem:footprint})};

  \begin{scope}[shift={(0, 1.9)}]
    \def\R{1.8}         
    \def\Rmid{2.15}     
    \def\Rout{2.65}     

    \draw[thick, fill=gray!5] (0,0) circle (\R);
    \draw[dashed, thin, color=gray!40] (0,0) circle (\Rout);

    \coordinate (Y1) at (90:\R);
    \coordinate (Y2) at (15:\R);
    \coordinate (Y3) at (280:\R);
    \coordinate (Y4) at (180:\R);

    \draw[draw=none, fill=gray!15] (0,0) -- (90:\R) arc (90:15:\R) -- cycle;

    \draw[draw=none, pattern=north east lines, pattern color=black!60] 
      (0,0) -- (90:\R) arc (90:40:\R) -- cycle;

    \draw[very thick] (90:\R) arc (90:15:\R);
    \draw[densely dashed, thick] (90:\R) arc (90:40:\R);

    \draw[thick] (0,0) -- (90:\R);
    \draw[thick] (0,0) -- (15:\R);
    \draw[thin, dashdotted] (0,0) -- (280:\R);
    \draw[thin, dashdotted] (0,0) -- (180:\R);

    \filldraw[fill=black] (Y1) circle (2.2pt) node[above=2pt, font=\small] {$y_i$};
    \filldraw[fill=black] (Y2) circle (2.2pt) node[right=2pt, font=\small] {$y_{i+1}$};
    \filldraw[fill=gray!60, draw=black] (Y3) circle (2pt) node[below=2pt, font=\small] {$y_{i+2}$};
    \filldraw[fill=black] (Y4) circle (2.2pt) node[left=2pt, font=\small] {$y_{i-1}$};

    \draw[->, thick] (160:1.15) arc (160:125:1.15) 
      node[midway, above left=-1pt, font=\scriptsize] {Scan $(+1)$};

    \draw[<->, thick] (80:\Rmid) arc (80:25:\Rmid);
    \node[font=\small] at (52:\Rmid+0.25) {$g_i$};

    \draw[<->, thin] (172:\Rmid) arc (172:98:\Rmid);
    \node[font=\small, text=black!80] at (135:\Rmid+0.22) {$g_{i-1}$};

    \node[font=\footnotesize, text=gray!70] at (-35:\Rmid) {$g_{i+1}$};

    \node[font=\scriptsize, inner sep=1pt] at (55:0.95) {$\min(g_i, t)$};

    \draw[->, thick, dashed] (180:\Rout) arc (180:55:\Rout);
    \node[font=\scriptsize, anchor=south east] at (165:\Rout+0.05) {Arc $A_{y_{i-1}, t}$};
    \draw[thick] (55:\Rout-0.12) -- (55:\Rout+0.12);
    \node[anchor=south west, font=\scriptsize, inner sep=1pt] at (55:\Rout+0.05) {$\le y_i+t-1$};

    \node[font=\footnotesize, text=gray!60] at (0,-0.25) {$\mathbb{Z}_N$};
  \end{scope}

  \draw[thin, color=gray!30] (3.1, 0.2) -- (3.1, 4.6);

  \node[anchor=south west, font=\normalsize\bfseries] at (4.0, 4.7) 
    {(b) Water-Filling Root $T^*$ (Theorems~\ref{thm:ttl} \& \ref{thm:lru})};

  \begin{scope}[shift={(4.4, 0.2)}]
    \draw[->, thick] (0,0) -- (6.0, 0) node[right=2pt, font=\small] {Sorted Gaps $g_{(i)}$};
    \draw[->, thick] (0,0) -- (0, 3.8);
    \node[rotate=90, font=\small, anchor=south] at (-0.55, 1.9) {Blocks / Window $t$};

    \def\W{0.7}
    \def\gA{0.6}
    \def\gB{1.1}
    \def\gC{1.6}
    \def\gD{2.7}
    \def\gE{3.3}
    \def\Tstar{2.0}

    \fill[gray!15, draw=none] (0, \Tstar-0.4) rectangle (5.4, \Tstar+0.4);
    \draw[dashed, thin, color=gray!70] (0, \Tstar-0.4) -- (5.4, \Tstar-0.4) 
      node[right=2pt, font=\tiny, text=black!70] {$\Ts-1$ (Hit)};
    \draw[dashed, thin, color=gray!70] (0, \Tstar+0.4) -- (5.4, \Tstar+0.4) 
      node[right=2pt, font=\tiny, text=black!70] {$\Ts+1$ (Miss)};


    \draw[very thick, black] (0, \Tstar) -- (5.4, \Tstar) 
      node[right=2pt, font=\footnotesize\bfseries] {$\Ts = \lfloor\theta\rfloor$};

    \fill[pattern=north east lines, pattern color=black!60] (0.3, 0) rectangle (0.3+\W, \gA);
    \draw[thick] (0.3, 0) rectangle (0.3+\W, \gA);
    \node[below=2pt, font=\scriptsize] at (0.3+0.5*\W, 0) {$g_{(1)}$};

    \fill[pattern=north east lines, pattern color=black!60] (1.2, 0) rectangle (1.2+\W, \gB);
    \draw[thick] (1.2, 0) rectangle (1.2+\W, \gB);
    \node[below=2pt, font=\scriptsize] at (1.2+0.5*\W, 0) {$g_{(2)}$};

    \fill[pattern=north east lines, pattern color=black!60] (2.1, 0) rectangle (2.1+\W, \gC);
    \draw[thick] (2.1, 0) rectangle (2.1+\W, \gC);
    \node[below=2pt, font=\scriptsize] at (2.1+0.5*\W, 0) {$g_{(3)}$};

    \fill[pattern=north east lines, pattern color=black!60] (3.1, 0) rectangle (3.1+\W, \Tstar);
    \fill[gray!30] (3.1, \Tstar) rectangle (3.1+\W, \gD);
    \draw[thick] (3.1, 0) rectangle (3.1+\W, \gD);
    \draw[thick, line width=1.1pt] (3.1, \Tstar) rectangle (3.1+\W, \gD);
    \node[below=2pt, font=\scriptsize] at (3.1+0.5*\W, 0) {$g_{(4)}$};
    \node[font=\footnotesize\bfseries] at (3.1+0.5*\W, \Tstar+0.35) {Miss};

    \fill[pattern=north east lines, pattern color=black!60] (4.0, 0) rectangle (4.0+\W, \Tstar);
    \fill[gray!30] (4.0, \Tstar) rectangle (4.0+\W, \gE);
    \draw[thick] (4.0, 0) rectangle (4.0+\W, \gE);
    \draw[thick, line width=1.1pt] (4.0, \Tstar) rectangle (4.0+\W, \gE);
    \node[below=2pt, font=\scriptsize] at (4.0+0.5*\W, 0) {$g_{(5)}$};
    \node[font=\footnotesize\bfseries] at (4.0+0.5*\W, \Tstar+0.65) {Miss};

    \draw[decorate, decoration={brace, amplitude=4pt}, thick] (3.1, 3.5) -- (4.0+\W, 3.5);
    \node[above=3pt, font=\scriptsize\bfseries] at (3.9, 3.5) {$K(\Ts) = 2$ Misses};

    \node[anchor=north west, font=\scriptsize, fill=white, draw=gray!40, rounded corners=2pt, inner sep=2.5pt] at (0.1, 3.7) {
      \begin{tabular}{@{}ll@{}}
        \tikz{\fill[pattern=north east lines, pattern color=black!60] (0,0) rectangle (0.22, 0.11); \draw (0,0) rectangle (0.22, 0.11);} & $U(\Ts) \le C$ \\
        \tikz{\fill[gray!30] (0,0) rectangle (0.22, 0.11); \draw[thick] (0,0) rectangle (0.22, 0.11);} & Uncached gap $> \Ts$\\
        \tikz{\fill[gray!15] (0,0) rectangle (0.22, 0.11); \draw[dashed] (0,0.05) -- (0.22,0.05);} & LRU ambiguity band
      \end{tabular}
    };
  \end{scope}

\end{tikzpicture}
\caption{Address ring and water-filling. (a) The phases partition the
ring into territories of lengths $g_i$; a window of $t$ rounds covers
$\min(g_i,t)$ blocks of each (Lemma~\ref{lem:footprint}). (b) The sorted
gaps filled to capacity $C$ give the waterline $\theta$ and window
$\Ts=\lfloor\theta\rfloor$; each gap above it costs one miss per round
(Theorem~\ref{thm:ttl}). In the band $[\Ts-1,\Ts+1]$ LRU may differ
from TTL (Theorem~\ref{thm:lru}).}
\label{fig:static_geometry}
\end{figure*}
Figure~\ref{fig:static_geometry}(a) shows the phases on the address ring,
with the gaps as the arcs between consecutive distinct phases. The results
below depend on the phases only through the gaps: where the workers stand
on the ring does not matter, only how far apart they are. 
\subsection{Exact traffic for fixed phases: the gap law}
\label{sec:ttl}
Fix the phases. How many blocks must each round fetch from memory? We
write this count as $K$, the misses per round, so that $K/M$ is the miss
rate per request. The answer has a water-filling form
(Figure~\ref{fig:static_geometry}(b)): sort the gaps, raise a common
waterline until exactly $C$ blocks lie beneath it, and count the gaps
that rise above it. That number is $K$. Lemma~\ref{lem:footprint} and
Theorem~\ref{thm:ttl} make this exact.

The footprint $U(t)$ is the number of distinct blocks in any $t$
consecutive complete rounds. It is independent of the starting round,
because advancing all phases merely rotates the address ring.

\begin{lemma}[Cyclic-gap footprint]
\label{lem:footprint}
For every integer $t\ge0$,
\begin{equation}
 U(t)=\sum_{i=1}^{D}\min(g_i,t),\qquad
 K(t):=U(t+1)-U(t)=\#\{i:g_i>t\}.
 \label{eq:footprint}
\end{equation}
\end{lemma}

Extend $U$ to real $t\ge0$ by $U(t)=\sum_i\min(g_i,t)$.
Since $C<N$, the equation $U(\theta)=C$ has a unique solution $\theta$, the
waterline of Figure~\ref{fig:static_geometry}(b).
Sorting the gaps $g_{(1)}\le\cdots\le g_{(D)}$ gives
\begin{equation}
 \theta=\frac{C-\sum_{i=1}^{j}g_{(i)}}{D-j},\qquad
 g_{(j)}\le\theta<g_{(j+1)},\qquad \Ts=\lfloor\theta\rfloor.
 \label{eq:waterfill}
\end{equation}
Here $0\le j<D$, $g_{(0)}=0$, and an empty sum is zero.
Equivalently, $T^*$ is the largest integer round window whose
footprint fits within the boundary capacity:
\[
    U(T^*)\le C<U(T^*+1).
\]
Sorting and a prefix scan cost $O(M\log M)$, independent of trace length.

We use a precisely slotted TTL abstraction, the object Denning called the working-set policy~\cite{denning1} and the caching literature calls a TTL cache~\cite{berger2014}: at the start of a round,
the cache remembers the preceding $T$ complete rounds; current-round
duplicates share one lookup/fill. At the next boundary it retains the
latest $T$ complete rounds. Thus boundary occupancy is $U(T)$.
Temporary current-round fills are not part of this boundary capacity
contract. A request-level TTL may need additional within-round storage;
it is not identified with a strict capacity-$C$ hardware cache.

\begin{theorem}[Exact slotted miss count]
\label{thm:ttl}
In steady state, the $\Ts$-window cache has boundary occupancy at most
$C$ and exactly
\begin{equation}
 F_C(\boldsymbol g)=K(\Ts)=\#\{i:g_i>\Ts\},\qquad
 \rho_{\rm TTL}=F_C(\boldsymbol g)/M
 \label{eq:ttl}
\end{equation}
misses per round and misses per request, respectively. Current-round fills lie outside this boundary capacity; Theorem~\ref{thm:lru} ties the count to capacity-$C$ LRU.
\end{theorem}

Since $C<N$, at least one gap exceeds $\Ts$. Therefore, the best TTL count is
$K=1$, attained by \emph{aligned} workers (one phase) and by tight
groups of \emph{separated} workers (distinct phases). If all gaps exceed $C/D$, then
$\theta=C/D$ and $K=D$, and each distinct phase pays its own fill.
Section~\ref{sec:dynamic} takes $K=1$ as its anchor and predicts accumulated
traffic directly, because separated workers can still reach $K=1$.
\subsection{LRU differs from TTL only at the waterline}
Hardware caches are not slotted TTL caches. The next theorem shows that
for fully associative LRU the gap law stays exact except for gaps within
one round of the waterline, so the closed form carries over to LRU
without a trace.
\begin{theorem}[A one-round LRU certificate]
\label{thm:lru}
Assume $D\le C$. For fully associative LRU in steady state,
\begin{equation}
 K(\Ts+1)\le K_{\rm LRU}\le K(\Ts-1).
 \label{eq:lru}
\end{equation}
Only phases with preceding reuse lag $g_i\in\{\Ts,\Ts+1\}$ can disagree
with the TTL classification.
\end{theorem}

When its interval collapses, the inexpensive TTL formula also gives
exact fully associative LRU. Otherwise one can compute each worker's
exact $\mathrm{RD}$ by merging the at-most-$M$ cyclic address intervals
between its previous and current request, in $O(M^2\log M)$ total time.
This symbolic alternative needs no generated or measured address trace.
It is also the reuse profile given to the SDCM--FA baseline in
Section~\ref{sec:exp}.


\subsection{A policy-independent traffic lower bound}
\label{sec:policy}
LRU is only the most common policy. Could a smarter one avoid the traffic that separated phases create? Theorem~\ref{thm:policy} bounds every demand-paging policy’s traffic from below using only the phases.
Preserve within-round order by placing worker $w$'s round-$k$ request at time $k+w/M$. Let $\tau_w>0$
be its elapsed time since the preceding request to the same block.
In rounds, $\tau_w=\min_{w'}[g_{w'w}+(w-w')/M]$ with $g_{w'w}=(s_{w'}-s_w)\bmod N$,
a zero lag counting as $N$ when $w'\ge w$, so $\sum_w\tau_w=N$; two aligned workers give $(N-\tfrac12,\tfrac12)$.

\begin{theorem}[Policy-independent residency bound]
\label{thm:policy}
For any capacity-$C$ demand-paging policy without prefetching, including
offline Belady~\cite{belady1966}, the long-run miss count satisfies
\begin{equation}
 K_{\rm policy}\ge M-\Phi(C),\qquad
 \Phi(C)=\max_{\substack{0\le x_w\le1\\
                     \sum_w\tau_wx_w\le C}}\sum_w x_w .
 \label{eq:policy}
\end{equation}
For an $H$-round finite trace with arbitrary initial cache, the lower
bound has an additive correction of at most $N/H$.
\end{theorem}

The program is fractional knapsack: sort $\tau_w$, retain all cheapest
intervals that fit, and fractionally fill the next one. It is a fractional interval-packing relaxation, akin to LP relaxations of offline and generalized caching~\cite{albers1999,barnoy2001}. It is an optimistic
relaxation, not a realizable policy. Unlike a bound based only on the
smallest gap, it accounts for the limited number of cheap reuse events.
Using exact $\tau_w$, rather than integer $g_i$, also avoids incorrectly
discarding request-order boundary effects. We evaluate it against Belady
in Section~\ref{sec:exp}.

For $N=8,M=2,C=2$ and phases $(0,4)$, Theorem~\ref{thm:policy} gives $K_{\rm policy}\ge10/7>1$, certifying excess traffic even against the aligned LRU count $K=1$.
Sections~\ref{sec:coherence} and~\ref{sec:feedback} bound how far the workers may separate before sharing is lost, and show why the cache does not pull them back together.
\subsection{How much lead the cache tolerates}
\label{sec:coherence}
The fixed-phase analysis determines traffic for given relative scan
positions. We now ask how far workers can fall behind one another without
losing shared cache fills. Let $M$ workers repeatedly read blocks
$0,1,\ldots,N-1$, all starting at block zero but advancing at their own
pace. For cache analysis, arrange their requests into a single event
sequence that preserves each worker's request order. The event index $u$
counts requests across all workers, not elapsed time or scan rounds.
Let $n_w(u)$ be the number of requests made by worker $w$ among the first
$u$ events, with $n_w(0)=0$. Its next request has cumulative index
$k=n_w(u)$ and accesses block $k\bmod N$. These counts do not reset
between scans, so a worker one full scan behind is $N$ requests behind,
even if both workers are at the same block. Over the first $E>0$ events,
define $\sigma_E$ as the largest lead of any worker over any other worker,
measured in requests:
\begin{equation}
 \sigma_E=\max_{0\le u\le E}
       \left(\max_w n_w(u)-\min_w n_w(u)\right).
 \label{eq:progress_spread}
\end{equation}
\begin{theorem}[A sharp capacity condition for preserving sharing]
\label{thm:coherence}
Let $M\ge2$ workers share an initially empty, fully associative LRU cache
of integer capacity $1\le C<N$, with no other accesses. If
$C\ge2\sigma_E-1$, every request to an index $k$ after its first
request by any worker hits. Consequently, the miss count $A_E$ satisfies
\begin{equation}
 A_E\le\max_w n_w(E),\qquad
 \frac{M A_E}{E}\le1+\frac{(M-1)\sigma_E}{E}.
 \label{eq:coherent_traffic}
\end{equation}
For every $\sigma\ge2$, a two-worker schedule with $\sigma_E=\sigma$ can lose
sharing at $C=2\sigma-2<N$; the threshold is sharp.
\end{theorem}

Informally, if no worker ever leads another by more than $\sigma$
requests, a capacity of $2\sigma-1$ blocks preserves the sharing
entirely, and a capacity of $2\sigma-2$ can already lose it. The proof needs $\sigma$ only within each reuse interval; other traffic there adds its distinct blocks to the required capacity.

Bounded progress differences thus preserve near-$K=1$ traffic over long
runs. The bound concerns the entire history: final phase alignment can
hide earlier separation or a whole-scan lag. It establishes when sharing
survives, but does not ensure that workers remain close.

\subsection{When hits cannot pull a slow worker back}
\label{sec:feedback}
One might expect divergence to correct itself: a worker that falls behind
finds its blocks already fetched by the leaders, hits more often, and
catches up. But sharing a fill and recovering a progress deficit are
different benefits.
Our TMA kernel primes $d-1$ copies into $d\ge2$ buffers, then repeatedly
\emph{waits for block $k$, consumes it, and issues block $k+d-1$}.
With two buffers, the next copy is issued only after the issuing warp's
current-block consumer instructions. Furthermore, the transaction barrier requires
both the data and all participating threads to arrive~\cite{cuda_async}:
a faster copy need not release a barrier held by a late warp.

For an indefinitely continued interior loop, put $r=d-1$ and
let $I_k$ be copy $k$'s issue time. Use fixed costs $\alpha\ge0$ from
the previous issue to the next wait, and $\beta>0$ from that wait's
release through consumption to the refill issue. Assume thread arrivals
add no later gate, and let $\lambda_k$ be copy latency including queuing.
The resulting earliest-event recurrence~\cite{baccelli1992} is
\begin{equation}
 I_{k+r}=\max\{I_{k+r-1}+\alpha,\ I_k+\lambda_k\}+\beta,
 \quad k\ge0,
 \label{eq:slot_timing}
\end{equation}
with finite initial issue times $I_0,\ldots,I_{r-1}$.

\begin{theorem}[A limit on cache-mediated catch-up]
\label{thm:feedback}
For constant latency $\lambda$, Equation~\eqref{eq:slot_timing} gives
\begin{equation}
 P_d(\lambda):=\lim_{n\to\infty}\frac{I_n}{n}
 =\max\left\{\alpha+\beta,\frac{\lambda+\beta}{d-1}\right\}.
 \label{eq:refill_cadence}
\end{equation}
More generally, allow arbitrary latency sequences between bounds
$0\le\lambda_{\rm hit}\le\lambda_k\le\lambda_{\rm miss}$. For two workers, let $P_w$
denote Equation~\eqref{eq:refill_cadence} with worker $w$'s costs and depth,
and give their latency bounds subscript $w$.
If $\phi=P_2(\lambda_{{\rm hit},2})-P_1(\lambda_{{\rm miss},1})>0$, then
\begin{equation}
 I_{2,n}-I_{1,n}\ge n\phi-O(1).
 \label{eq:catchup_limit}
\end{equation}
Even the fastest permitted responses for worker 2 and the slowest for
worker 1 cannot repair this growing timing deficit.
\end{theorem}


Taking $\lambda_{\rm hit},\lambda_{\rm miss}$ as hit/miss latencies, define the available
overlap $\Theta=(d-1)(\alpha+\beta)-\beta$ and $[v]_+=\max(v,0)$.
The greatest asymptotic time saving per request from replacing every
miss by a hit is
\begin{equation}
 G_d=P_d(\lambda_{\rm miss})-P_d(\lambda_{\rm hit})
 =\frac{[\lambda_{\rm miss}-\Theta]_+-[\lambda_{\rm hit}-\Theta]_+}{d-1}.
 \label{eq:catchup_gain}
\end{equation}
When misses fit within $\Theta$, $G_d=0$, and any number of hits gives at most a bounded startup advantage.
Losing this correction alone does not create divergence;
Equation~\eqref{eq:catchup_limit} also requires a sustained service imbalance.
For $d=2$, $\Theta=\alpha$, so pre-refill consumer work provides no
overlap; every additional buffer adds one $\alpha+\beta$ interval.
Increasing only miss latency cannot decrease $G_d$. Therefore, explaining
weaker correction under pressure requires more than ``misses get slower'';
non-memory service and inter-warp arrival delays must also be examined.

A measured reference depth sweep makes this distinction concrete. At
$(N,M,C)=(4096,48,1280)$, 16\,KiB blocks and no added compute,
$d=2,3,4$ give first-wave times of 841, 465, and 424\,ns per normalized
round, all with $K\simeq1$; the fourth stage saves only 9\%, so by then
the copy latency is largely hidden and $G_d$ is near zero. Over 32
waves, $d\le3$ stays at $K=1.00$ while $d=4$ rises to 5.58 at an
unchanged 425\,ns per round, and eight WMMA
operations per block at $d=4$ keep $K=1.00$. One pipeline stage past the
point where latency is hidden thus trades a small time saving for a
fivefold traffic increase, and compute per block undoes it. Traffic and
timing come from separate runs, so this contrast shows that they
decouple and gives no estimate of $G_d$.

Section~\ref{sec:dynamic} predicts in exactly the variables that appear
here: prefetch depth ($r=d-1$ in the recurrence), the load path (only
TMA has the barrier gate above, so $\alpha$ and $\beta$ differ by path),
compute per block (the non-memory service in $\alpha$ and $\beta$,
coupled to $d$ through $\Theta$), and occupancy (the number of resident
workers $M$ in Theorem~\ref{thm:coherence} and in $\rho=K/M$). 
Finite-run phase evolution remains unresolved by these theorems; Section~\ref{sec:dynamic} instead calibrates the resulting fill traffic at reference configurations.
In a separate $W=32$ GB10 marker diagnostic, sampled within-batch lead stays below 110 blocks for all 96 depth-three CTA batches, while 94 of 96 depth-four batches cross the 640-block scale in Theorem~\ref{thm:coherence} at $C=1280$.
On marked $d=3/4$ runs, snapshot gap-law $K$ is $1.00/6.15$ versus whole-run fill $K=1.00/5.52$.
At $d=4,W=32$ on GB10, a 512-block within-wave checkpoint cuts measured $K$ from 5.40 to 1.00 (5.58, 5.52, and 5.40 are from separate runs) with 2.9\% plain kernel-time overhead, demonstrating how restricting lead at Theorem~\ref{thm:coherence}'s scale can restore sharing.
\section{Dynamic Miss-Rate Prediction}

\label{sec:dynamic}
Section~\ref{sec:static} therefore separates the prediction problem into two parts: relative progress determines a sharing geometry, and the sharing geometry determines traffic. The second mapping is characterized analytically; the first is history-dependent and cannot in general be reconstructed from static code or a terminal phase vector. PASCAL therefore calibrates the accumulated effect of the first mapping while retaining the analytical aligned-traffic anchor and structural coordinates identified above.
It also supplies what a
predictor can use before execution: the coordinates $(o,\pi,d,m)$ that
control relative progress, the aligned TTL reference $K = 1$, and the
normalization $\rho=K/M$. 
Theorem~\ref{thm:feedback} gives a conditional mechanism
for persistent divergence, while the aligned gap law supplies the traffic reference used below.
Divergence enters through the measured traffic of references. We measure the traffic of a reusable set of reference
configurations once per device and transfer it to an unmeasured
configuration along static program structure, in three steps: fit a
cumulative-traffic curve for each reference configuration
(Section~\ref{sec:curve}), encode the traffic relative to that reference
(Section~\ref{sec:excess}), and interpolate that encoding to the target
within its execution class (Section~\ref{sec:transfer});
Algorithm~\ref{alg:predict} lists these steps for one query.
Section~\ref{sec:protocol} states the calibration protocol that keeps
test configurations out of every fitting step, and
Section~\ref{sec:operand-sharing} extends the same steps to GEMM, whose
operands are only partially shared. Inference reads the frozen
calibration, the configuration, and the requested horizon; it uses no
timing, counters, or traces of the target. The result is a cache
component that a pre-execution performance model can call in place of
its uniform-advance reuse profile, which Section~\ref{sec:e2e} does
inside TileSight. Proofs and derivations are in
Appendix~\ref{app:sec4}.

\subsection{The prediction target}
\label{sec:target}
Each worker is a GPU thread block (CTA) scanning $N$ data blocks, and
$M$ workers can be resident concurrently. Let $W$ be the total number
of CTA scans divided by $M$, so $H=WN$ is work per resident-worker
equivalent; rounds need not be synchronized, and CTAs need not run at
equal speed. If $V(H)$ bytes are fetched into L2 from off-chip memory
during the run, define
\begin{equation}
 A(H)=V(H)/b,\qquad
 \overline K_H=A(H)/H,\qquad
 \rho_H=\overline K_H/M ,
 \label{eq:dynamic_target}
\end{equation}
where $b$ is bytes per data block. $A$ counts fill-equivalent blocks,
and $\overline K_H$ generalizes the static $K$ to a whole-run average
that includes startup and tail traffic. Below, $K$ abbreviates
$\overline K_H$ when the configuration and horizon are understood.

We predict this integral rather than a phase vector or a terminal
state, for two reasons. Under variable progress the reuse distance of
an access is the union of the cyclic intervals that each worker
requested in between; it can be computed from an observed history, whereas the current phase vector alone does not determine future interleaving. And a
run shorter than the growth scale of any saturation curve identifies
only one combination of that curve's amplitude and scale, never the
plateau itself (Appendix~\ref{app:identifiability}).

\subsection{From measurements to a traffic curve}
\label{sec:curve}
Write $m=m_{\rm comp}$ for compute work per block and $z=(o,\pi,d,m)$
for a configuration key, where $o$ is resident CTAs per streaming
multiprocessor, $\pi$ selects the synchronous or Tensor Memory
Accelerator (TMA) load path, and $d$ is prefetch depth ($d=0$ for synchronous
loads). A family $f=(o,\pi,d)$ varies only $m$. Blocks are $b=16$~KiB;
$C=1280$ and $M=48o$ on GB10, $C=2048$ and $M=20o$ on Thor.

\begin{algorithm}[t]
\caption{Predict the miss rate of an unmeasured configuration}
\label{alg:predict}
\begin{algorithmic}[1]
\Require target key $z=(o,\pi,d,m)$, scan length $N$, scan count $W$;
         the frozen calibration
\State find the family $f=(o,\pi,d)$; reject if it has no references
\For{each reference size $N_j$ adjacent to $N$ in $\eta=C/N$}
  \For{each reference key $z'\in f$}
    \State $L_{z'}\gets$ curve of $z'$ at horizon $W$
           \Comment{\S\ref{sec:curve}}
    \State $a_{z'}\gets\log\big([\exp L_{z'}-1]_++\varepsilon\big)$
           \Comment{Eq.~\eqref{eq:excess_coordinate}}
  \EndFor
  \State $p\gets$ PCHIP of $\{a_{z'}\}$ over $m$, evaluated at the target $m$
  \State $a_L,a_R\gets$ nearest class-$\chi(m)$ references around $m$
  \State $\widehat a_j\gets(1-\omega)p+\omega a_{\rm lin}$
         \Comment{Eq.~\eqref{eq:graph}}
  \State \quad if no same-class bracket: fall back across $m$, then $d$
\EndFor
\State $\widehat a\gets$ linear interpolation of $\widehat a_j$ in $\eta$, clamped
\State \Return $\widehat K=\min\{M,1+[\exp\widehat a-\varepsilon]_+\}$,
       $\widehat\rho=\widehat K/M$ \Comment{Eq.~\eqref{eq:dynamic_decode}}
\end{algorithmic}
\end{algorithm}

Scan length $N$ and scan count $W$ are kept separate, because equal
total work $H=NW$ does not imply equal reuse: the static gap law depends
on $C/N$, and $W$ counts scan restarts. At each calibrated size $N_j$ we
fit and transfer curves over $W$, then interpolate the predictions
across $C/N$ (Section~\ref{sec:transfer}).

The curve is fitted to cumulative traffic rather than to $K$, and we
impose that cumulative traffic is nondecreasing across the
independently launched horizons. At each calibrated horizon
$W_j\in\{1,4,16,64\}$, we take the geometric mean of the repeated
measurements, make the logarithm of cumulative traffic nondecreasing by
isotonic regression, and interpolate it over $\log W$ with a
shape-preserving piecewise cubic Hermite interpolant
(PCHIP)~\cite{fritsch1984}. Subtracting $\log W$ gives
$L_z(W)\approx\log K$ at the requested horizon; outside the calibrated
range, $K$ is held at the nearest endpoint (Appendix~\ref{app:curvefit}).

\subsection{Encoding traffic relative to aligned scans}
\label{sec:excess}
The Section~\ref{sec:feedback} phase-snapshot check links divergence to fills through the gap law; when $1\le C<N$, PASCAL interpolates measured excess over its coincident-phase $K=1$ count.
With the positive part $[v]_+=\max(v,0)$ of Section~\ref{sec:feedback}, encode each curve as
\begin{equation}
 a_z(W)=\log\!\left([\exp L_z(W)-1]_++\varepsilon\right),
 \qquad \varepsilon=0.01 .
 \label{eq:excess_coordinate}
\end{equation}
Figure~\ref{fig:transfer}(a) plots this encoding. For a transferred value $\widehat a$, the prediction is
\begin{equation}
 \widehat K=\min\{M,\ 1+[\exp\widehat a-\varepsilon]_+\},
 \qquad \widehat\rho=\widehat K/M .
 \label{eq:dynamic_decode}
\end{equation}
Near alignment, this coordinate attenuates a given transfer error in
relative terms (Appendix~\ref{app:excessbound}). With PASCAL’s transfer unchanged, replacing this anchor-relative encoding by $\log K$ raises the four-panel balanced MAPE from 12.04\% to 17.05\%.

\subsection{Transfer to an unmeasured configuration}
\label{sec:transfer}
Lines 7--10 of Algorithm~\ref{alg:predict} transfer the encoded curves
to the target compute; this subsection explains their two choices, the
execution class and the anchored interpolation.

\paragraph{Why execution classes.}
Compute count alone hides changes in generated control flow. On both
binaries, disassembly shows that the loop tail depends on $m\bmod4$,
and remainders two and three run the same tail
(Appendix~\ref{app:obstruction}). We therefore define the execution
class $\chi(m)=\min(m\bmod4,\,2)$ and assume smoothness of the
encoded response only within a class.
Both the compute argument and the disassembly are available before
execution. Every class therefore needs at least one reference: without one,
the class is unidentifiable, because two
response laws can agree on every reference and still differ by a
factor up to $M$ on it (Appendix~\ref{app:obstruction}).

\begin{figure}[t]
\centering
\includegraphics[width=\columnwidth]{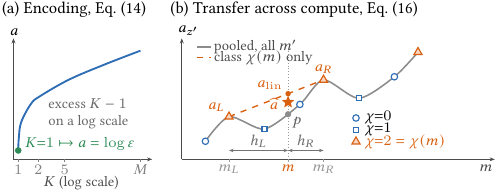}
\caption{Encoding and transfer, on illustrative data. (a)~Log excess over the aligned count $K=1$. (b)~Equation~\eqref{eq:graph}: the prediction $\widehat a$ lies between the pooled trend $p$ and the same-class line $a_{\rm lin}$.}
\label{fig:transfer}
\end{figure}

\paragraph{Anchored interpolation within a class.}
Fix a family $f$ and horizon $W$. First interpolate $a_z(W)$ across all
reference compute values, irrespective of class (PCHIP with at least
three distinct knots, linear otherwise, clamped outside the range), and
write $p$ for the pooled estimate at the query $m$. If the nearest
same-class references bracket $m$ at $m_L<m<m_R$, with encoded values
$a_L,a_R$ and distances $h_L=m-m_L$, $h_R=m_R-m$, the transferred value
is a convex combination of the pooled trend and linear interpolation
within the class,
\begin{equation}
 \begin{aligned}
 \widehat a&=(1-\omega)p+\omega a_{\rm lin},\qquad
 a_{\rm lin}=\frac{h_Ra_L+h_La_R}{h_L+h_R},\\
 \omega&=\frac{\gamma(1/h_L+1/h_R)}{c_0+\gamma(1/h_L+1/h_R)},
 \end{aligned}
 \label{eq:graph}
\end{equation}
where $c_0=(m-m_-)^{-1}+(m_+-m)^{-1}$ uses the adjacent reference
computes $m_-<m<m_+$ without class restriction. This is anchored
harmonic interpolation~\cite{zhu2003}: Equation~\eqref{eq:graph}
minimizes a quadratic energy that penalizes departure from the pooled
trend and squared slope along the two same-class edges
(Appendix~\ref{app:energy}); Figure~\ref{fig:transfer}(b) draws it. The single weight $\gamma$ sets how
strongly same-class references override the pooled trend; we keep
$\gamma=4$ from development without retuning it. Queries without a
same-class bracket fall back to endpoint-clamped interpolation across
compute and then prefetch depth; unsupported occupancy/path families
are rejected rather than extrapolated.

\paragraph{Across scan sizes.}
The class-level predictions at the two reference sizes adjacent to $N$
are interpolated linearly in $\eta=C/N$ and clamped outside the
calibrated range (line~12 of Algorithm~\ref{alg:predict}); both devices
use the same rule and reference-size grid.

\subsection{Calibration protocol and error control}
\label{sec:protocol}
Calibration uses one fixed, reusable collection of configurations and
horizons, and nothing from a test configuration enters any fitting
step. In each development fold, remove a complete key $z=(o,\pi,d,m)$
at every $N,W$, including its short-run measurements. Refit all curves
and execution-class transfers at every reference size using only the
remaining references. For validation keys $\mathcal Z$, let $n_z$ be
the number of scored runs for key $z$, and $K_{zj}>0,\widehat K_{zj}$
their observed and out-of-fold predicted counts. Select by
configuration-balanced MAPE:
\begin{equation}
 \frac{100}{|\mathcal Z|}
 \sum_{z\in\mathcal Z}\frac{1}{n_z}
 \sum_{j=1}^{n_z}
 \left|\frac{\widehat K_{zj}}{K_{zj}}-1\right|.
 \label{eq:dynamic_selection}
\end{equation}
This is miss-rate MAPE because the same $M$ divides both counts. Both
devices use the same 63 calibration depth/compute coordinates and 756-label grid (Section~\ref{sec:dynamic_exp}), selected for coverage without test
traffic. Whole-key folds validate the fixed rule; they do not select a
new rule or weight. Learning baselines receive the same within-device
labels and static features. All forecasts are frozen before either
test panel is acquired and are not refitted afterwards.
Table~\ref{tab:contract} (Appendix~\ref{app:contract}) summarizes the resulting contract.

Proposition~\ref{prop:graph} (Appendix~\ref{app:certificate}) bounds
the transfer error along Equation~\eqref{eq:graph}: a denser compute
grid shrinks its curvature term, but not the error of noisy references.

\subsection{Extension to shared operands}
\label{sec:operand-sharing}
Partially shared operands require static address-sharing
geometry. Given reference bytes $0<B_0(H)\le MbH$, define
$K_0=B_0/(bH)$ and $\kappa=K/K_0$. Set $L_z=\log\kappa$ in
Equation~\eqref{eq:excess_coordinate} and decode as
$\widehat K=\min\{M,K_0(1+[\exp\widehat a-\varepsilon]_+)\}$.
The scan has $B_0=bH$, hence $K_0=1$.
Under the same
encoded-error condition and $\kappa\in[1,M/K_0]$, the relative-error
bound~\eqref{eq:excessbound} carries over because $K_0$ cancels.

For the tested FP16 GEMM, $B_0$ follows from the tile counts of the
output grid and the reduction length, and the transfer runs along the
two grid dimensions in turn (Appendix~\ref{app:gemm}); the three steps
of Sections~\ref{sec:curve}--\ref{sec:transfer} are reused unchanged,
with $B_0$ as the new static reference.

\subsection{Scope and limitations}
\label{sec:scope}
The gap law is exact for its abstract traces and
Proposition~\ref{prop:graph} holds under its stated conditions; the
curve fit, aligned $K = 1$ anchor, the class rule and the boundary rules are
modeling choices validated only within the calibrated scope of
Table~\ref{tab:contract}. Theorem~\ref{thm:feedback} models one mechanism read from the kernel
source and leaves the rest of the GPU schedule out. Calibration is device- and
operator-specific; unknown binaries and concurrent kernels are
unsupported.

\section{Experimental Evaluation}
\label{sec:exp}
\subsection{Experimental Setup and Overhead}
We measure the cache miss rate of fixed-phase kernels to test the theorems of Section~\ref{sec:static}, and of dynamic kernels with shared cyclic scans, such as GEMM, to test the prediction pipeline of Section~\ref{sec:dynamic}.
NVIDIA GB10 has 48 SMs and a 24MiB L2 cache, and Jetson Thor T5000 has 20 SMs and a 32 MiB cache. 
Referring to the capacity scan experiments shown in Figure~\ref {fig:capacity}, we choose 20MiB as the effective cache capacity for GB10 because significant off-chip traffic is observed when the scanned data size exceeds 20MiB. By contrast, Jetson Thor's effective capacity is equal to its nominal capacity. The data block size is 16KiB in dynamic experiments, so the effective capacity is $C=1280$ for GB10 and $C=2048$ for Jetson Thor, respectively.

The predicted metric is the fill-equivalent miss rate $\rho_H$ of Equation~\eqref{eq:dynamic_target}: DRAM-to-L2 fill bytes divided by logical read bytes. Unlike the lookup miss rate, which measures the fraction of L2 accesses that miss, our fill-equivalent miss rate measures DRAM-to-L2 fill bytes per logical operand-read byte. We use this traffic-oriented metric because it does not classify MSHR misses and MSHR hits into a single category. The metric quantifies the actual penalties of cache misses, providing a meaningful input to bandwidth-based performance models for ML kernels.

Calibration is conducted once per device and kernel family:
\CostGBRefRuns\ scan reference runs, \CostGBRefWallMin/\CostThorRefWallMin~minutes
of profiling on GB10/Thor. For the 42 held-out scan keys with $N=4096$ and
$W=64$, Nsight Compute restricted to the single fill counter PASCAL predicts
takes a median \NcuGBSec/\NcuThorSec~s per run on the device, whereas a
PASCAL forecast of the same runs takes \CostGBPredictMs/\CostThorPredictMs~ms
on average on one CPU core and needs neither the run nor the device.
Simulating one such run, occupancy one on the synchronous path, in
GPGPU-Sim~\cite{accelsim0} takes an estimated
\SimOneLow--\SimOneHigh~CPU-hours. A PASCAL forecast is thus three to four orders of magnitude faster than profiling and eight to nine orders of magnitude faster than cycle-accurate simulation.

\begin{figure}[t]
\centering
\includegraphics[width=\columnwidth]{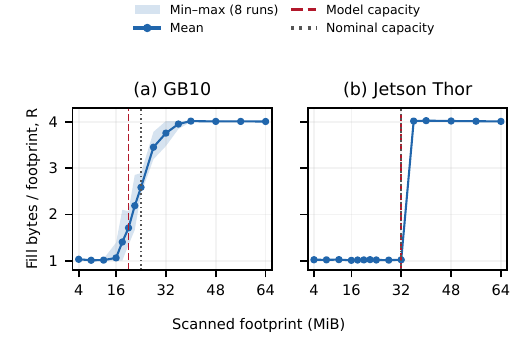}
\caption{Capacity scan: 15 sizes, eight repeats/device,
with the same synchronous kernel, 64 MiB allocation, 16 KiB
blocks, one CTA/SM.}
\Description{Off-chip fill volume normalized by scanned bytes increases
from approximately one at small sizes to nearly four at 32 MiB.
Vertical lines mark the 20 MiB model capacity and 24 MiB nominal capacity.}
\label{fig:capacity}
\end{figure}
\subsection{Static accuracy and replacement bounds}
\begin{figure*}[t]
\centering
\includegraphics[width=\textwidth]{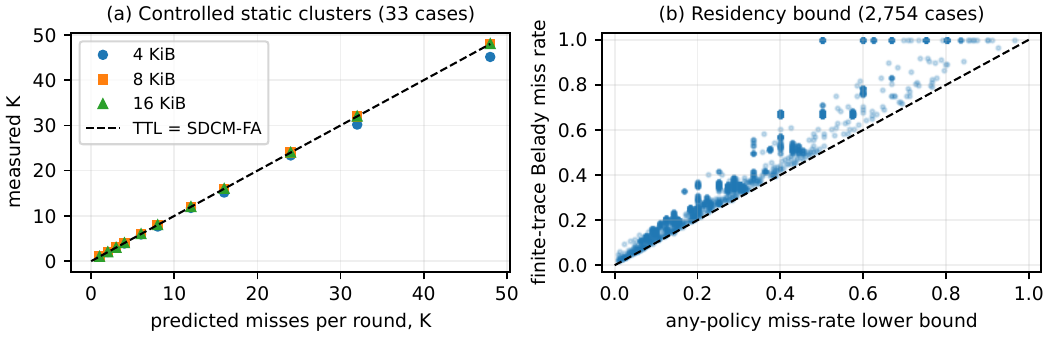}
\caption{GB10 static hardware accuracy and abstract-model validation.
(a) The first experiment (clusters), 33 configurations in Table~\ref{tab:static}. (b) Comparing our lower bound on cache misses with the miss count under optimal replacement
  (Belady).}
\label{fig:static}
\end{figure*}

\begin{table}[t]
\centering\small
\caption{GB10 static hardware miss-rate MAPE (\%).}
\label{tab:static}
\begin{tabular}{lrrr}
\toprule
Cases (\# of configurations) & PASCAL & SDCM--FA & Che--IRM\\
\midrule
Clusters: all sizes (33) & 1.19 & 1.19 & 632.96 \\
\midrule
Phase/size sweep (121) & 9.97 & 9.97 & 1124.63 \\
\bottomrule
\end{tabular}
\end{table}

Static experiments mean that the phase distribution among different workers does not change during the experiments. In actual kernels, it can happen in many cases, e.g., the kernel latency is short, or the pressure on the memory subsystem is low. In our micro-benchmarks, we use a grid-level barrier (\texttt{cg::grid\_group.sync()}) to ensure static conditions so that we can compare PASCAL with SDCM-FA and Che-IRM fairly, while more actual dynamic experiments without the grid-level barrier are conducted in Section~\ref{sec:dynamic_exp}. The SDCM-FA baseline is given the exact symbolic reuse distance profile rather than the hardware tracing. Algorithm~\ref{alg:static-kernel} presents the structure of the micro-benchmark kernel. 

We verify the accuracy of Equation~\eqref{eq:ttl} against two experimental settings on GB10. For the first experiment, we equally partition thread blocks into different clusters, and the thread blocks inside the same cluster have the same phase. The phases are equally spaced. We test 33 cases with $M$ = 48, $N$ = 8192, 4000 rounds, 4/8/16 KiB blocks, and 11 cluster counts from 1 through 48.

The second experiment uses various phase patterns (e.g. Wrapped Gaussian, Uniform Random, and Equally Spaced) and distinct $N/C$ and $M$ values with 2 KiB blocks. The results of all 154 cases are presented in Table~\ref{tab:static}. Across these 154 static configurations, Equation~\eqref{eq:ttl} and SDCM-FA yield the same MAPE, and they significantly outperform Che--IRM. The results verify our analysis in Section~\ref{sec:background}: Che--IRM does not capture phase-dependent cache sharing between thread blocks.
Therefore, the gap law matches the exact-distance baseline at $O(M\log M)$ cost; the residual error reflects hardware effects neither model captures.

To verify the bounds of Theorems~\ref{thm:lru} and~\ref{thm:policy}, we simulate every case with $N=2,\ldots,6$, $M=1,\ldots,3$,
all phase vectors, and $1\le C<N$, which gives 2254 cases; 500 seeded larger cases bring the total to 2754. LRU miss counts lie within the predicted bounds whenever $D\le C$. Figure~\ref{fig:static}(b) shows that even optimal (Belady) replacement incurs no fewer misses than our policy-independent lower bound, and no case violates it. The figure uses forced-allocation MIN; against MIN with bypass, the mean slack falls from 0.11 to 0.07. 
The bound isolates the traffic that remains under every demand-paging replacement policy at the stated capacity.

\begin{algorithm}[t]
\caption{Static phase-controlled scan (one thread block $w$)}
\label{alg:static-kernel}
\begin{algorithmic}[1]
\Require array of $N$ logical blocks, phase $s_w$, rounds $H$
\State $a\gets0$; each thread owns one 16-byte lane
\For{$k=0,\ldots,H-1$}
  \State $j\gets(k+s_w)\bmod N$
  \State all threads load their lanes of block $j$ with \texttt{\_\_ldcg}
  \State consume the loaded values into $a$
  \State grid-wide barrier across all $M$ thread blocks
\EndFor
\State write $a$ to prevent dead-code elimination
\end{algorithmic}
\end{algorithm}

\subsection{Why dynamic prediction is necessary}

\begin{figure*}[t]
\centering
\includegraphics[width=\textwidth]{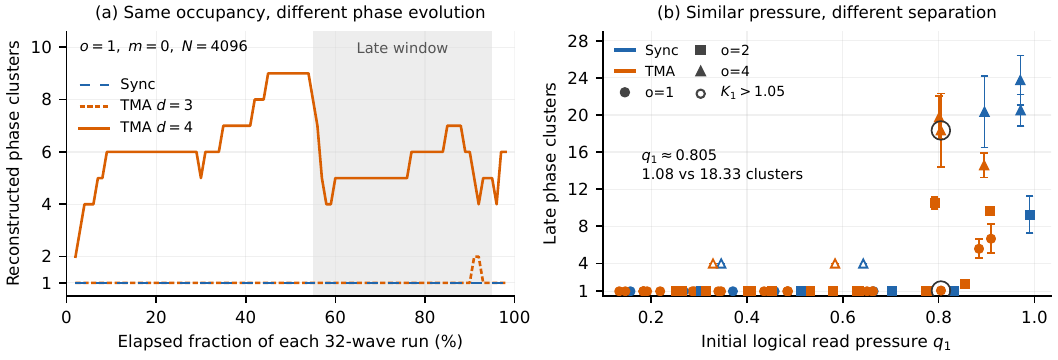}
\caption{GB10 phase separation.
(a) Same occupancy, different evolution.
(b) Initial pressure versus late separation.
Worker positions are estimated from start/end times; error bars show variation over time.}
\Description{Three matched one-CTA-per-SM trajectories show synchronous
and depth-three TMA scans near one cluster, but depth-four TMA separating.
The complete sixty-configuration pressure plot retains four previously
omitted cases as hollow markers. Two highlighted TMA configurations have
pressure about 0.805 but late mean cluster counts of 1.08 and 18.33.}
\label{fig:phenomena}
\end{figure*}
\begin{algorithm}[t]
\caption{Open scan with a binary-defined priming count}
\label{alg:dynamic-kernel}
{
\begin{algorithmic}[1]
\Require $N$, path $\pi$, prefetch depth $d$, compute work $m$
\Require priming count $n_{\rm p}$ fixed by the binary: $n_{\rm p}=d-1$ on GB10, $n_{\rm p}=d$ on Thor
\State each of $WM$ CTAs starts at block zero; no grid barrier
\State if $\pi$ is TMA, issue blocks $0,\ldots,\min(n_{\rm p},N)-1$
\For{$j=0,\ldots,N-1$}
  \If{$\pi$ is TMA}
    \State issue block $j$ if not yet issued; wait for its data
  \Else
    \State cooperatively load block $j$ synchronously
  \EndIf
  \State consume block $j$
  \State execute $m$ iterations of the compute loop
  \State perform required intra-CTA synchronization
  \If{$n_{\rm p}>0$ and $j+n_{\rm p}<N$}
    \State issue block $j+n_{\rm p}$ into its available buffer
  \EndIf
\EndFor
\State write results and optional diagnostic markers
\end{algorithmic}}
\end{algorithm}
Actual kernels do not usually use grid-level barriers as we
do in the static experiments, and progress divergence among workers can
happen in such cases. To understand progress divergence, we consider
different techniques frequently used in kernel optimization in our
benchmark kernel (Algorithm~\ref{alg:dynamic-kernel}). The kernel can
either use Tensor Memory Accelerator (TMA) to fetch data, which is similar
to DMA in NPUs, or synchronous loads (through \texttt{\_\_ldcg()}).
The depth of software pipelines, the FLOPs per block, and occupancy are
also variables in our controlled experiments. These four variables form the
configuration key $z=(o,\pi,d,m)$ of Section~\ref{sec:curve}.

We define the logical read pressure of a configuration as
\begin{equation}
 q_1=\frac{Mb/\bar t_1}{B_{\rm ref}},
 \qquad B_{\rm ref}=2.10\ {\rm TB/s},
 \label{eq:pressure}
\end{equation}
where $\bar t_1$ is the mean time per work-normalized round in a $W=1$
reference run and $B_{\rm ref}$ is an empirical read-rate reference.
The ratio counts logical reads, including hits and merged requests, and
is not an L2-to-SM utilization counter; neither $\bar t_1$ nor a pressure
threshold is an input to the predictor. The subscript 1 marks this $W=1$
run throughout: $K_1$ is its whole-run count $\overline K_H$
(Equation~\eqref{eq:dynamic_target}).

We define a phase cluster by merging workers whose rounded
positions are connected through cyclic gaps of at most 32 blocks.
In separate $W=32$ GB10 checks, eight sparse markers per CTA confirm about one versus seven clusters at $d=3,4$; median runtime changes by at most 1.8\%, and the separation persists for 16--128-block merge thresholds.
Progress divergence can be summarized by the number of phase clusters.
When progress divergence occurs, the phases of different workers can form
more clusters. Figure~\ref{fig:phenomena}(a) shows the complexity of the
problem; occupancy alone cannot determine whether progress divergence
will occur, since the TMA kernel with $d=4$ suffers from progress
divergence when occupancy is 1, while the synchronous kernel
does not. Figure~\ref{fig:phenomena}(b) compares phase separation with
initial memory pressure evaluated by Equation~\eqref{eq:pressure}.
Among configurations with $K_1\le1.05$, greater separation accompanies
higher initial memory pressure in this sweep. However, we cannot
quantitatively calculate the divergence from memory pressure alone,
because kernels at similar pressure can present different degrees
of divergence.

\subsection{Dynamic prediction: comparisons}
\label{sec:dynamic_exp}
\begin{figure*}[t]
\centering
\includegraphics[width=\dimexpr\textwidth-2\fboxsep-2\fboxrule\relax]{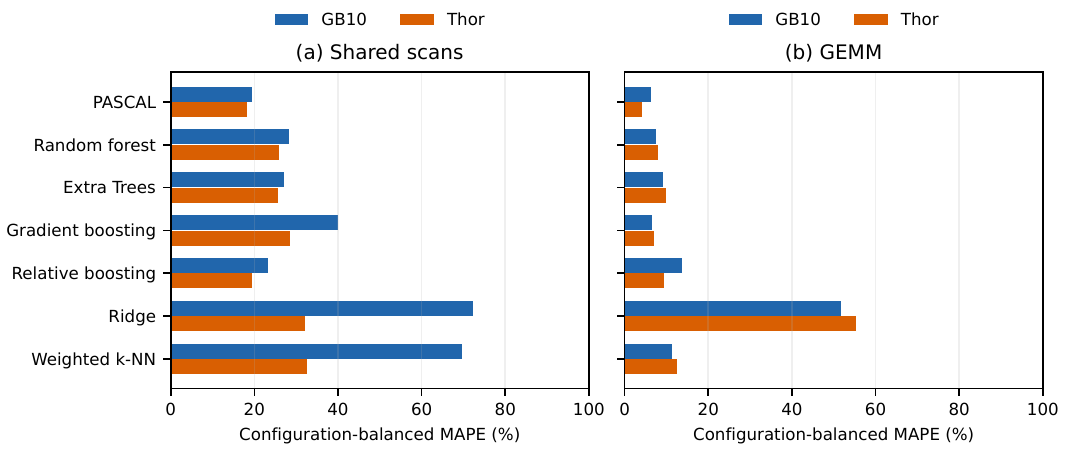}
\caption{The same learning-baseline comparison on both
workloads and devices, with a common error axis. Each method uses the same
calibration labels and held-out outcomes within a device--workload group.
Scans use 756 labels per device; GEMM uses 120 on each. Relative boosting uses absolute-error gradient boosting with inverse observed-traffic weights.
Table~\ref{tab:literature-cache} reports the literature comparisons and
their equally weighted aggregate.}
\label{fig:prediction}
\end{figure*}

\begin{table}[t]
\centering\small\setlength{\tabcolsep}{3pt}
\begin{tabular}{@{}llrrrr@{}}
\toprule
Workload & Device & PASCAL & TileSight & PPT-GPU & SDCM\\
\midrule
Scan & GB10 & 19.41 & 63.10 & 78.70 & 29.24\\
Scan & Thor & 18.14 & 23.55 & 78.26 & 23.21\\
GEMM & GB10 & 6.41 & 29.12 & 26.18 & 28.21\\
GEMM & Thor & 4.22 & 27.14 & 27.40 & 24.07\\
\midrule
\multicolumn{2}{@{}l}{Balanced overall} & 12.04 & 35.73 & 52.63 & 26.18\\
\bottomrule
\end{tabular}

\caption{Miss-rate MAPE (\%) decomposed
by workload and device.}
\label{tab:literature-cache}
\end{table}

We test two categories of kernels here: shared cyclic scan kernels with different settings shown in Algorithm~\ref{alg:dynamic-kernel} and CUTLASS GEMM kernels~\cite{Thakkar_CUTLASS_2023}. The synthetic kernels highlight the shared cyclic scan memory access pattern and cover a wide range of occupancies and prefetch depth, while CUTLASS GEMM kernels test the usability of PASCAL in production environments.

For the shared cyclic scan kernels, we test 42 depth/compute pairs with the highest feasible occupancy settings at each depth on both devices.
The 42 held-out keys $(o,\pi,d,m)$ belong to these calibrated kernel families and enter neither model selection nor calibration.
Each key is tested with different workload sizes, from short-work pairs (e.g., $N = 4096, W = 1$) to long-work pairs (e.g., $N = 5120, W = 80$). For the GEMM kernels ($M_g \times N_g \times K_g$), 60 different shapes are tested, with $M_g\in\{640,896,1280,2560,5120\}$,
$N_g\in\{640,896,1280,2560\}$ and $K_g\in\{1280,2560,5120\}$. We hold out these shapes for testing, while using 120 different shapes with $M_g\in\{256,\ldots,8192\}$, $N_g\in\{256,\ldots,4096\}$, and $K_g\in\{1024,\ldots,8192\}$ for calibration. The tile shape is $128\times128\times64$, and the precision is FP16 for inputs and FP32 for outputs. Both binaries fix a direct CTA raster (GB10 identity swizzle; Thor direct grid mapping), and 25/60 GB10 versus 40/60 Thor test shapes fit both FP16 inputs plus FP32 output within effective L2.

Table~\ref{tab:literature-cache} compares our model with the cache components in TileSight~\cite{tilesight}, PPT-GPU~\cite{PPTGPU}, and exact SDCM. All three SDCM-based cache model columns share one hit rule, a reuse hits if fewer than eight of its intervening lines map to its set, and differ only in inputs and arithmetic: TileSight takes reuse distances from its own tile schedule, PPT-GPU and exact SDCM from our symbolic scan profile, and only exact SDCM evaluates the binomial hit probability without a normal approximation. PASCAL achieves a MAPE of 12.04\%, which outperforms all three baselines. We also compared PASCAL's prediction pipeline with six machine-learning-based methods in Figure~\ref{fig:prediction} under the same training set and the same testing set. Across scan test runs, $K>1.05$ occurs in 145/420 on GB10 and 239/420 on Thor.
Each key runs twice per $(N,W)$, and predicting each run by its repeat already gives 5.62/13.90\% on GB10/Thor scans. 

Our exploratory multi-query FlashAttention-2 forward implementation uses 96 reference and 24 held-out runs per device at disjoint sequence lengths, where Thor $K$ spans 1.04--7.71 across depths 1--3; query (Q)-adjusted shared key/value (K/V) traffic gives PASCAL 1.96/11.44\% MAPE on GB10/Thor versus 12.06/40.88\% for the aligned $K=1$ baseline.
Each CTA computes a 64-query tile using $QK^\top$, online softmax, and $PV$; sampled forward outputs agree with a CPU reference to within $5.5\times10^{-6}$.
On 15 further held-out $(N,W,d)$ cells per device (two repeats each), the same fit gives 1.00/11.82\% MAPE on GB10/Thor versus 14.33/45.58\% at $K=1$.
\subsection{End-to-end latency with TileSight}
\label{sec:e2e}
\begin{table}[t]
\centering\small\setlength{\tabcolsep}{3pt}
\caption{Latency MAPE (\%) with TileSight, PASCAL, or measured L2 fills as the pipeline's cache input.}
\label{tab:e2e}
\begin{tabular}{@{}lllrrr@{}}
\toprule
& & & \multicolumn{3}{c}{Cache input to the pipeline}\\
\cmidrule(l){4-6}
Workload & Device & Runs & TileSight & PASCAL & Measured\\
\midrule
Scan & GB10 & all & 23.89 & 23.89 & 23.89\\
Scan & Thor & all & 76.21 & 76.21 & 76.20\\
GEMM & GB10 & all & 18.17 & 13.14 & 10.92\\
GEMM & GB10 & DRAM-bound & 42.06 & 19.85 & 10.11\\
GEMM & Thor & all & 81.84 & 81.70 & 81.54\\
\bottomrule
\end{tabular}

\end{table}
To test whether the Table~\ref{tab:literature-cache}
difference holds in a complete performance prediction pipeline (Table~\ref{tab:e2e}), we change only the
cache input of TileSight's pipeline and wave model~\cite{tilesight}:
its own cache component, PASCAL, or measured fills (the pipeline's error
floor). Its equations, microbenchmarked peaks, the held-out keys and
PASCAL's calibration are shared; a scan block is one 16~KiB load and $m$
matrix instructions over $\max(d,1)$ stages. Nothing is fitted to a
duration; latency is the kernel time from the same profiling pass as each
fill counter.

A run is DRAM-bound when measured fills increase the pipeline's predicted time by more than 2\% versus zero DRAM traffic; the intervals below are 95\% paired key-bootstrap confidence intervals.
On GB10 GEMM, PASCAL lowers latency MAPE from
18.17\% to 13.14\%, and from 42.06\%
to 19.85\% on its 27\ DRAM-bound runs
(paired-gap intervals [\EtoEGemmGBGapLow, \EtoEGemmGBGapHigh] and
[\EtoEGemmGBBoundGapLow, \EtoEGemmGBBoundGapHigh] points; measured-fill input
\EtoEGemmGBBoundMeasured\%). Elsewhere DRAM time hides below L2 or compute
time, so all inputs agree even when miss-rate errors differ (63.10\%
versus 19.41\% on GB10 scans): as in Equation~\eqref{eq:catchup_gain},
fewer fills buy no time once miss service fits within the overlap. On Thor, even measured fills leave median predicted time at only 0.175/0.182 of observed scan/GEMM time, locating the dominant error outside the cache input.
\section{Related Work}
\subsection{Footprints, distances, and statistical caching}
Denning’s working set theory~\cite{denning1, denning2},  Mattson’s stack analysis~\cite{stack}, and HOTL~\cite{hotl} are the conceptual foundations of our static geometry analysis. We reinterpret the working set theory in our terms and specialize it as the gap formula for AI workloads. Such a simplification drives our research into more complex dynamic analysis. Che–IRM and its temporal-locality extensions~\cite{che1, che2, che3} differ mainly in the request laws supplied to characteristic-time analysis. StatCache~\cite{statcache} and StatStack~\cite{statstack} reduce profiling cost through statistical locality measurements. Similarly, whole-program reuse prediction~\cite{whole} and static array-profile estimation~\cite{array} show that pre-execution reuse distance prediction is possible.
Database cooperative scans preserve buffer reuse by grouping, throttling, or rescheduling concurrent readers~\cite{lang2006,zukowski2007}; our fixed-schedule setting instead bounds cache traffic from their relative progress.
\subsection{Multicore and GPU models}
Cycle-accurate simulators, such as gpgpu-sim~\cite{accelsim0}, accel-sim~\cite{AccelSim, accelsim2}, and Sim-FA~\cite{simfa} model the memory subsystem pipeline in detail. Therefore, while high accuracy is possible through these methodologies, the simulation speed tends to be extremely low. This challenge is even more severe in the Large Language Models (LLM) era due to the large problem size and large-scale design space exploration (DSE) with these tools is not practical in such a context. Nugteren et al.'s reuse-distance GPU cache model~\cite{nugteren2014} and OWL's CTA-aware warp scheduling~\cite{jog2013} target complementary prediction and scheduling choices.
\subsection{Parallel locality}
Blelloch and Gibbons~\cite{Blelloch} prove shared-cache scheduling guarantees relative to sequential cache performance, using additive cache augmentation tied to parallelism and computation depth. Chen et al.~\cite{chen} demonstrate constructive sharing with parallel depth-first scheduling and study task granularity on CMPs. These papers establish why scheduling is a locality resource. 
Theorem~\ref{thm:coherence} addresses a different guarantee. Rather than choosing a schedule and comparing its cache complexity with a sequential execution, it conditions directly on the realized progress lead \(\sigma\) and gives a sharp, schedule-independent \(2\sigma-1\) capacity threshold for preserving shared fills. Thus the earlier results establish that scheduling can create locality, whereas our result quantifies exactly how much relative progress a cyclic shared scan can tolerate under an arbitrary interleaving.
\section{Conclusion}
Shared-cache reuse in parallel cyclic scans is governed not only by what each worker accesses, but by how far their progress is allowed to separate. PASCAL formalizes this connection: cyclic gap geometry determines traffic for fixed relative progress, while a sharp \(2\sigma-1\) capacity threshold characterizes when arbitrary progress interleavings still preserve sharing.
Dynamic GPU execution does not expose future progress statically, so the predictor calibrates the accumulated consequence of this latent evolution rather than attempting to reconstruct it.
Under its stated recurrence, Theorem~\ref{thm:feedback} identifies when hidden miss latency prevents cache hits from repairing a timing gap. Prefetch depth, compute work, and occupancy index the calibrated predictor, whose balanced held-out MAPE is 12.04\%.

\section*{Acknowledgments}
We acknowledge the use of Claude Code and OpenAI Codex for language editing and code implementation. No AI tools were used to generate research ideas, analyses, or experimental results. All code and text were reviewed, tested, and validated by the authors, who take full responsibility for the contents of this work.
\bibliographystyle{ACM-Reference-Format}
\bibliography{sample-base}
\clearpage
\appendix
\section*{Appendix}
\section{Notation}
\label{app:notation}
Table~\ref{tab:notation} lists the symbols that recur across sections;
symbols used only near their definition are omitted. Workers are indexed
by $w$, distinct phases and gaps by $i$, and rounds or requests by $k$;
$m$ always denotes compute work per block.

\begin{table}[H]
\centering\footnotesize
\caption{Recurring notation and where each symbol is defined.}
\label{tab:notation}
\begin{tabular}{@{}l p{.62\columnwidth} l@{}}
\toprule
Symbol & Meaning & Def.\\
\midrule
$N,\ C$ & Blocks per scan; cache capacity in blocks & \S\ref{sec:model}\\
$M$ & Concurrently resident workers & \S\ref{sec:model}\\
$w,\ s_w$ & Worker index; phase of worker $w$ & \S\ref{sec:model}\\
$D,\ g_i$ & Number of distinct phases; $i$-th gap & \S\ref{sec:model}\\
$\Ts$ & TTL window at capacity $C$ & Eq.~\eqref{eq:waterfill}\\
$K,\ \rho$ & Misses (fills) per round; miss rate $K/M$ & \S\ref{sec:ttl}\\
$\sigma_E$ & Largest lead of any worker over another & Eq.~\eqref{eq:progress_spread}\\
$d$ & Prefetch depth & \S\ref{sec:feedback}\\
$G_d$ & Largest time saving per request from hits & Eq.~\eqref{eq:catchup_gain}\\
\midrule
$o,\ \pi$ & Resident CTAs per SM; load path (sync or TMA) & \S\ref{sec:curve}\\
$m$ & Compute work per block ($m_{\rm comp}$) & \S\ref{sec:curve}\\
$z,\ f$ & Key $(o,\pi,d,m)$; family $(o,\pi,d)$ & \S\ref{sec:curve}\\
$W,\ H$ & CTA scans per resident worker; $H=WN$ & \S\ref{sec:target}\\
$\overline K_H$ & Whole-run average of $K$ & Eq.~\eqref{eq:dynamic_target}\\
$L_z$ & Fitted $\log K$ curve of key $z$ & \S\ref{sec:curve}\\
$a_z,\ \varepsilon$ & Encoded excess over $K=1$; floor $0.01$ & Eq.~\eqref{eq:excess_coordinate}\\
$\chi(m)$ & Execution class, $\min(m\bmod4,2)$ & \S\ref{sec:transfer}\\
$\omega,\ \gamma$ & Within-class weight; its strength, $4$ & Eq.~\eqref{eq:graph}\\
$\eta$ & $C/N$ & \S\ref{sec:transfer}\\
\midrule
$K_0,\ \kappa$ & Static reference count; $K/K_0$ & \S\ref{sec:operand-sharing}\\
$M_g,N_g,K_g$ & GEMM dimensions (not $M$, $N$, $K$ above) & App.~\ref{app:gemm}\\
\bottomrule
\end{tabular}
\end{table}

\section{Proofs for Section~\ref{sec:static}}
\label{app:sec3}

\begin{proof}[Proof of Lemma~\ref{lem:footprint}]
For $0\le t<N$, partition the address ring into the territories $[y_i,y_{i+1})$. Within territory $i$, the nearest preceding phase in the scan direction is $y_i$. Hence a block is covered by a length-$t$ scan segment exactly when it lies among the first $\min(g_i,t)$ blocks of that territory. For $t\ge N$, every block is covered and the same formula equals $N$. Summing over the territories gives $U(t)$; taking the finite difference gives $K(t)$.
\end{proof}

\begin{proof}[Proof of Theorem~\ref{thm:ttl}]
At each distinct phase, only the first current-round request can miss.
Its previous occurrence was requested by the next phase $g_i$ rounds
earlier (or $N$ rounds earlier for a single phase). It is remembered
exactly when $g_i\le\Ts$. Feasibility follows from $U(\Ts)\le C$.
\end{proof}

\begin{proof}[Proof of Theorem~\ref{thm:lru}]
Let $\mathrm{RD}$ count distinct blocks strictly between two successive
references to the same block, separated by $g>0$ rounds.
The intervening sequence contains $g-1$ complete rounds and is contained
in $g+1$ rounds including the referenced block. Therefore
$U(g-1)\le\mathrm{RD}\le U(g+1)-1$.
LRU hits iff $\mathrm{RD}<C$. A lag at most $\Ts-1$ forces a hit,
and a lag at least $\Ts+2$ forces a miss. Current-round duplicates hit
because at most $D-1<C$ distinct blocks intervene.
\end{proof}

\begin{proof}[Proof of Theorem~\ref{thm:policy}]
A hit with its previous reference inside the horizon requires uninterrupted
residency throughout that inter-reference interval: demand paging cannot
reload an unrequested, evicted block. Intervals of the same block have
disjoint interiors, and at most $C$ blocks are resident at a time.
Thus their total length is at most $CH$. If $Hx_w$ counts such hits of
worker type $w$, then $x_w\le1$ and $\sum_w\tau_wx_w\le C$.
At most $N$ additional hits can have their preceding reference outside
the horizon. Divide by $H$ and let $H$ increase.
\end{proof}

\begin{proof}[Proof of Theorem~\ref{thm:coherence}]
\label{p:coherence}
Write $\sigma=\sigma_E\ge1$. Immediately after the first request to
index $k$, every worker has completed at least $k+1-\sigma$ requests.
Until another worker requests $k$, that worker's count is at most $k$;
after every intervening event, all counts are therefore at most
$k+\sigma$. Thus all intervening request indices lie in
$[k+1-\sigma,k+\sigma-1]$. Excluding $k\bmod N$ leaves at most
$2\sigma-2$ distinct addresses. The same bound holds since its most
recent reference, so Mattson's LRU criterion~\cite{stack} forces a hit.
There are only $\max_w n_w(E)$ first indices. Also,
$E=\sum_w n_w(E)\ge M\max_w n_w(E)-(M-1)\sigma$, proving the bounds.

For sharpness, fix $\sigma\ge2$ and choose $N=2\sigma$. Worker 1 first requests $0,\ldots,\sigma-1$;
worker 2 requests $0,\ldots,\sigma-2$; worker 1 requests $\sigma,\ldots,2\sigma-2$;
worker 2 then requests $\sigma-1$. The largest count difference is $\sigma$.
Between the two references to $\sigma-1$ there are exactly $2\sigma-2$ distinct
other addresses, so the latter reference misses at capacity $2\sigma-2$.
\end{proof}

\begin{proof}[Proof of Theorem~\ref{thm:feedback}]
Write $c=\alpha+\beta$ and $P=\max\{c,(\lambda+\beta)/r\}$.
The recurrence has a one-index edge of weight $c$ and an $r$-index
edge of weight $\lambda+\beta$. Iterating each edge gives
$I_n\ge nP-O(1)$. For $B=\max_{0\le j<r}(I_j-jP)$,
induction gives $I_n\le nP+B$, since both edge weights are at most
their index lengths times $P$. Hence $I_n=nP+O(1)$.
Monotonicity sandwiches every variable-latency clock between its
constant-bound clocks, proving Equation~\eqref{eq:catchup_limit}.
\end{proof}

\section{Derivations for Section~\ref{sec:dynamic}}
\label{app:sec4}

\subsection{A short profile does not identify a saturation law}
\label{app:identifiability}
Assume an instantaneous fill count
$K_{\rm inst}(u)=1+\mathcal A[1-\exp(-(u/\ell)^\nu)]$ at work coordinate $u$,
with amplitude $\mathcal A\ge0$, scale $\ell>0$, and exponent $\nu>0$.
Averaging $K_{\rm inst}$ over $0\le u\le H$ gives, as $H/\ell\to0$,
\begin{equation}
 \overline K_H-1=
 \frac{\mathcal A}{\nu+1}(H/\ell)^\nu
 +O\!\left(\mathcal A(H/\ell)^{2\nu}\right).
 \label{eq:identifiability}
\end{equation}
Short profiles primarily identify $\mathcal A/\ell^\nu$, not the plateau
and growth scale separately. This is why Section~\ref{sec:target} fits
the observed integral directly, without assuming this saturation law
or solving for an unobserved equilibrium.

\subsection{Fitting the traffic curve}
\label{app:curvefit}
Section~\ref{sec:curve} describes the fit in words; here it is as
formulas. At each distinct reference horizon $W_j$, aggregate positive measured
counts into their geometric mean $\widetilde K_z(W_j)$, project
$Y_j=\log W_j+\log\widetilde K_z(W_j)$ onto nondecreasing values by least
squares (isotonic regression), and interpolate over $\log W_j$ with a
shape-preserving piecewise cubic Hermite interpolant
(PCHIP)~\cite{fritsch1984}, denoted $S_z$. The untransferred curve is
\begin{equation}
 L_z^0(W)=S_z(\log W)-\log W ,
 \label{eq:trafficcurve}
\end{equation}
an estimate of $\log K$ at the requested horizon. The calibration
grid supplies $W=1,4,16,64$ at every reference key and size, so
$L_z=L_z^0$ and no synthetic point is needed.

\subsection{Relative error of the excess coordinate}
\label{app:excessbound}
If the true count is
$K=1+e\in[1,M]$, $a=\log(e+\varepsilon)$, and $|\widehat a-a|\le\delta$,
then for $0<\varepsilon\le1$
\begin{equation}
 \frac{|\widehat K-K|}{K}
 \le \frac{e+\varepsilon}{1+e}(\exp(\delta)-1)
 \label{eq:excessbound}
\end{equation}
Near alignment the prefactor attenuates a given coordinate error; how
small that error is depends on the transfer
(Appendix~\ref{app:certificate}).
\begin{proof}
Positive-part projection is 1-Lipschitz, and clipping to $M$ cannot
increase error for $K\le M$. Before clipping, the excess error is at
most $(e+\varepsilon)|\exp(\widehat a-a)-1|$, which gives
Equation~\eqref{eq:excessbound}.
\end{proof}

\subsection{Execution classes and the unsampled-class obstruction}
\label{app:obstruction}
For the GB10 scan binary, disassembly of the matrix loop gives
$m=16q+8b_8+4b_4+(m\bmod4)$, with $q\ge0$ full unrolled groups
and $b_8,b_4\in\{0,1\}$ residual groups. Remainder
zero skips the tail, one exits early, and two and three traverse the
same tail blocks with the third matrix instruction predicated, which gives
$\chi(m)=\min(m\bmod4,\,2)$. Thor's four-iteration
unrolled tensor loop has the same tail partition, so one class rule
serves both binaries without equating their instructions or timings.

\begin{proposition}[An unsampled-class obstruction]
\label{prop:obstruction}
Fix a family and horizon, and assume smoothness only within each
execution class. If a target class has no references, two response
laws can agree on all references yet give target counts $1$ and
$\Lambda$, where $1<\Lambda\le M$. Every point predictor then has worst-case
relative error at least $(\Lambda-1)/(\Lambda+1)$ over these laws.
\end{proposition}
\begin{proof}
Choose the laws constant within the unmeasured class and identical
elsewhere. Both satisfy any within-class smoothness bound. For a
prediction $v$, the minimum of
$\max\{|v-1|,|v-\Lambda|/\Lambda\}$ is attained at $v=2\Lambda/(\Lambda+1)$, giving the bound.
\end{proof}

\subsection{Anchored interpolation as an energy minimizer}
\label{app:energy}
With $p$, $a_L,a_R$, $h_L,h_R$, and $c_0$ as in
Section~\ref{sec:transfer}, Equation~\eqref{eq:graph} is the minimizer
over trial coordinates $v$ of
\begin{equation}
 \mathcal E(v)=c_0(v-p)^2+
 \gamma\left\{\frac{(v-a_L)^2}{h_L}
             +\frac{(v-a_R)^2}{h_R}\right\}.
 \label{eq:graphenergy}
\end{equation}
Each edge penalty is the integral of squared slope along a linear
segment. Strict convexity yields the unique minimizer
\[
 \widehat a=
 \frac{c_0p+\gamma a_L/h_L+\gamma a_R/h_R}
 {c_0+\gamma/h_L+\gamma/h_R},
\]
which rearranges to Equation~\eqref{eq:graph}. All weights are
nonnegative, so the result lies in the convex hull of $p,a_L,a_R$.

\subsection{Conditional interpolation certificate}
\label{app:certificate}
\begin{proposition}[Conditional interpolation certificate]
\label{prop:graph}
Fix a reference size, family, execution class, and horizon. Assume its true encoded
response has a twice-differentiable extension $a_*(m)$ on the reference
bracket of width $h=h_L+h_R$, with $|a_*''|\le J$.
Suppose $|p-a_*(m)|\le\delta_0$ at the query and each endpoint estimate
has error at most $\xi$. Then
\begin{equation}
 |\widehat a-a_*(m)|
 \le(1-\omega)\delta_0+\omega(\xi+Jh^2/8)=:\Delta .
 \label{eq:graphbound}
\end{equation}
At this size, for a target count in $[1,M]$, Equation~\eqref{eq:excessbound}
applies with $\delta=\Delta$.
\end{proposition}
\begin{proof}
The linear-interpolation remainder is at most $Jh^2/8$.
Its nonnegative weights preserve the endpoint error bound $\xi$.
Apply the triangle inequality to Equation~\eqref{eq:graph};
no independence between the estimates is required.
\end{proof}
Between adjacent reference coordinates $\eta_L<\eta_R$, the transferred
value is $\widehat a=(1-\psi)\widehat a_L+\psi\widehat a_R$ with
$\psi=(\eta-\eta_L)/(\eta_R-\eta_L)$.
For interpolation between sizes, suppose the true encoded
response has $|\partial_\eta^2 a_*|\le J_\eta$ on $[\eta_L,\eta_R]$.
With endpoint error bounds $\Delta_L,\Delta_R$ from the proposition,
the same linear-interpolation remainder gives
$\delta\le(1-\psi)\Delta_L+\psi\Delta_R+
J_\eta(\eta-\eta_L)(\eta_R-\eta)/2$ in Equation~\eqref{eq:excessbound}.
This conditional bound does not cover clamped extrapolation, and finite differences of the references do not certify $J$ or $\delta_0$.
Here $J,\delta_0,\xi$ bound within-class curvature, pooled prediction
error, and reference-curve error, respectively. The latter includes
measurement noise and horizon-transfer approximation.
The result separates calibration spacing from reference uncertainty:
halving $h$ quarters the curvature term but does not reduce $\xi$.
For a purely linear within-class estimate and desired relative error
$\zeta>0$, the conservative condition
\[
 h\le\sqrt{\frac{8[\log(1+\zeta)-\xi]}{J}}
\]
suffices when $J>0$ and $\log(1+\zeta)>\xi$, since the prefactor in
Equation~\eqref{eq:excessbound} is at most one.

\subsection{Prediction contract}
Table~\ref{tab:contract} states what a prediction consumes and produces.
\label{app:contract}
\begin{table}[h]
\centering\small
\caption{Prospective prediction contract. A configuration key is
$z=(o,\pi,d,m)$; $H=WN$ is the requested horizon.}
\label{tab:contract}
\begin{tabular}{@{}p{.25\columnwidth}p{.68\columnwidth}@{}}
\toprule
Category & Content\\
\midrule
Calibration input & Reference configurations and horizons with measured
whole-run fill traffic\\
Target input & $(o,\pi,d,m,N,W)$ and the known binary's static
execution class\\
Prediction output & Fill-equivalent count $\widehat K_H$ and miss rate
$\widehat\rho_H=\widehat K_H/M$\\
Validated scope & Two devices (GB10 and Jetson Thor), two workloads
(the scan of Algorithm~\ref{alg:dynamic-kernel} and CUTLASS FP16 GEMM),
16 KiB logical blocks, and the calibrated configuration envelope\\
\bottomrule
\end{tabular}
\end{table}

\subsection{GEMM instantiation}
\label{app:gemm}
For the static reference of Section~\ref{sec:operand-sharing}, at fixed
scan size and constant $K_0$,
$\log(W\kappa)=\log V-\log(bNK_0)$,
so the cumulative fit differs only by translation.
For the tested FP16 GEMM, dimensions $(M_g,N_g,K_g)$
give tile counts $(n_{\rm r},n_{\rm c},n_{\rm k})=(M_g/128,N_g/128,K_g/64)$.
At reduction step $k$, output-tile CTA $(i,j)$ reads $A_{i,k}$ and $B_{k,j}$;
the $n_{\rm r}n_{\rm c}$ CTAs collectively read $n_{\rm r}+n_{\rm c}$ distinct 16 KiB tiles.
Thus $B_0=16384(n_{\rm r}+n_{\rm c})n_{\rm k}$ bytes, while logical input demand is
$Q=32768n_{\rm r}n_{\rm c}n_{\rm k}$. Set $b=32768$ bytes and $H=n_{\rm r}n_{\rm c}n_{\rm k}/M$,
giving $K_0=M(n_{\rm r}+n_{\rm c})/(2n_{\rm r}n_{\rm c})$ and $\rho=V/Q$.
$B_0$ is a cold-input reference, not simultaneous CTA residency;
both operands compete in the same measured response $V$.
At fixed $(n_{\rm r},n_{\rm c})$, fit the cumulative curve over reduction length $n_{\rm k}$.
Evaluate references at the requested $n_{\rm k}$, encode
$V/B_0$, then transfer along $\log n_{\rm c}$ followed by $\log n_{\rm r}$.
Within the single execution class, Equation~\eqref{eq:graph} with
unchanged weight 4 reduces to $(p+4a_{\rm lin})/5$ on each bracket;
endpoints are clamped. Path, occupancy, pipeline depth and per-step compute
are fixed for each tested binary; $(n_{\rm r},n_{\rm c})$ specifies sharing, not a new
measured feature. 

\end{document}